\documentclass[a4paper,11pt]{article}

\usepackage{complexity} 
\usepackage{enumerate}
\usepackage{amsmath, amssymb, amsthm,complexity}
\usepackage{euler}
\usepackage{authblk}
\usepackage{todonotes}
\usepackage{setspace}

\def\margin{2.9cm}
\usepackage[left=\margin,right=\margin,top=\margin,bottom=\margin]{geometry}

\usepackage{thmtools}
\usepackage{thm-restate}
\usepackage{hyperref}
\usepackage{cleveref}

\usepackage{authblk}
\usepackage{algorithm2e}

\theoremstyle{plain}
\newtheorem{theorem}{Theorem}
\newtheorem{definition}{Definition}[section]
\newtheorem{lemma}{Lemma}[section]
\newtheorem{corollary}{Corollary}[section]
\newtheorem{obs}{Observation}[section]
\newtheorem{claim}{Claim}[section]
\newtheorem{preprocess rule}{Preprocessing Rule}[section]
\newtheorem{reduction rule}{Reduction Rule}[section]
\newtheorem{branching rule}{Branching Rule}[section]
\usepackage{complexity}

\title{Tractability of {\Konig} Edge Deletion Problems\footnote{This work was done while the first author was affiliated to The Institute of Mathematical Sciences, HBNI, Chennai, India.}}

\date{}
\author[1]{Diptapriyo Majumdar}
\author[2]{Rian Neogi}
\author[2]{Venkatesh Raman}
\author[3]{S. Vaishali}
\affil[1]{Royal Holloway, University of London, Egham, United Kingdom\\
    \texttt{\{diptapriyo.majumdar@rhul.ac.uk\}}}
\affil[2]{The Institute of Mathematical Sciences, HBNI, Chennai, India\\
  \texttt{\{rianneogi|vraman\}@imsc.res.in}}
\affil[3]{PSG College of Technology, Coimbatore, India \texttt{\{svaishali.psg@gmail.com\}}}

\usepackage{todonotes}
\usepackage{microtype}
\usepackage{multirow}
\usepackage{amsmath,amssymb}

\usepackage{comment}
\usepackage{enumerate}

\usepackage{xspace}

\usepackage{listings}
\usepackage{color}
\usepackage{cite}
\usepackage{complexity}

\usepackage{graphicx}
\RequirePackage{fancyhdr}
 \usepackage{xcolor}
\usepackage{boxedminipage}

\newcommand{\NPComp}{NP-Complete}

\newcommand{\Wonehard}{W[1]-hard}

\newcommand{\KED}{{\sc K{\"{o}}nig Edge Deletion}}

\newcommand{\Konig}{{K{\"{o}}nig}}

\newcommand{\vc}{{\sf vc}}

\newcommand{\defproblem}[3]{
  \vspace{3mm}
\noindent\fbox{
  \begin{minipage}{.95\textwidth}
  \begin{tabular*}{\textwidth}{@{\extracolsep{\fill}}lr} #1  \\ \end{tabular*}
  {\bf{Input:}} #2  \\
  {\bf{Question:}} #3
  \end{minipage}
  }
  \vspace{2mm}
}

\begin{document}

\maketitle

\begin{abstract}
A graph is said to be a {\Konig} graph if the size of its maximum matching is equal to the size of its minimum vertex cover.
The {\KED} problem asks if in a given graph there exists a set of at most $k$ edges whose deletion results in a {\Konig} graph.
While the vertex version of the problem ({\Konig} vertex deletion) has been shown to be fixed-parameter tractable 
more than a decade ago, the fixed-parameter-tractability of the {\KED} problem has been open since then, 
and has been conjectured to be {\Wonehard} in several papers.
In this paper, we settle the conjecture by proving it {\Wonehard}.
We prove that a variant of this problem, where we are given a graph $G$ and a maximum matching $M$ and we want a $k$-sized {\Konig} edge deletion set that is disjoint from $M$, is fixed-parameter-tractable.
\end{abstract}

\section{Introduction}
A vertex cover in a graph is a subset of vertices such that it has at least one end point of every edge. A graph is {\Konig}-Egervary ({\Konig} in short) if its maximum matching size is equal to its minimum vertex cover size. 
The name {\Konig} comes from {\Konig}'s Theorem which states that the size of the minimum vertex cover is equal to the size of the maximum matching in bipartite graphs. It follows from {\Konig}'s Theorem that bipartite graphs are {\Konig}, but it is known that the converse is not true.

{\Konig} graphs, that form a super-class of bipartite graphs, have been well studied from structural and algorithmic perspectives~\cite{Stersoul1979ACO,Deming79}.
Given a graph, the problem of finding the minimum number of vertices or edges to be deleted so that the resulting graph satisfies a property, belongs to the class of vertex or edge deletion problems respectively. The vertex version of these problems include {\sc Vertex Cover} (where the resulting graph is an edgeless graph), {\sc Feedback Vertex Set} (where the resulting graph is a forest) and {\sc Odd Cycle Transversal} (where the resulting graph is bipartite). These problems are {\NPComp} by a celebrated result of Lewis and Yannanakis when the resulting graph class is a hereditary graph class (that is, it is closed under induced subgraphs)~\cite{LewisY80}. The parameterized complexity of these deletion problems have also been studied~\cite{cai1996fixed,lokshtanov2008wheel}.

Another natural version of these problems is to find a maximum induced subgraph satisfying the property.
There is a dichotomy theorem known in parameterized complexity for this version when the property is a hereditary property~\cite{KhotR02}. The problem has also been studied in the related notion of kernelization~\cite{KPRR14}.

The class of {\Konig} graphs, however, does not form a hereditary class and so the above results do not apply to it.
Approximation and fixed-parameter algorithms for deletion to or maximum induced {\Konig} graph were studied 
in~\cite{MRSSS11}. It was shown that it is fixed-parameter tractable (FPT) to determine whether a given graph has at most $k$ vertices whose deletion results in a {\Konig} graph. The running time of the algorithm was later 
improved~\cite{LNRRS14}.

The problem of finding a set of $k$ vertices that induce a subgraph that is {\Konig} has been proven to be {\Wonehard} while the related problem of finding a set of $k$ edges that induce a {\Konig} subgraph has been shown to be fixed-parameter-tractable. 
See ~\cite{konigkernel} for a recent result improving the FPT algorithm and 
the approximation ratio for the edge version, along with results on some `above-guarantee' version of the problem.

The parameterized complexity of the problem of deleting $k$ edges to obtain a {\Konig} graph has been open for 11 years~\cite{fpt_school} and is conjectured to be {\Wonehard}~\cite{MRSSS11}.
In this paper, we prove the conjecture positively, showing that $k$-edge deletion to {\Konig} graph is indeed {\Wonehard}.
We also show that a variant of the problem is fixed-parameter-tractable.

\section{Preliminaries}

Given a graph $G$, by $V(G)$, we denote the set of vertices of $G$, and by $E(G)$ we denote the set of edges. We represent an edge by a subset of $V(G)$ of size two. 
A matching is a set of edges that are pairwise disjoint. A vertex $v \in V(G)$ is said be an endpoint of an edge $e \in E(G)$ if $v \in e$. A vertex $v \in V(G)$ is said to be saturated by a matching $M$ if there exists an edge $e \in M$ such that $v$ is an endpoint of $e$, otherwise $v$ is said to be unsaturated by $M$. A set $S \subseteq V(G)$ is said to be saturated by a matching $M$ if $M$ saturates every vertex in $S$.
A vertex cover of a graph is a set of vertices whose removal from the graph results in an edgeless graph. By $\mu(G)$ we denote the size of a maximum matching of the graph and by ${\vc}(G)$ we denote the size of a minimum vertex cover of the graph.

\begin{definition}
A graph $G$ is said to be {\Konig} if $\mu(G) = {\vc}(G)$.
\end{definition}

Given two subsets of vertices $A, B \subseteq V(G)$, by $(A,B)$, we denote the set of edges that have one endpoint in $A$ and the other in $B$. Let $F \subseteq E(G)$, by $G[F]$ we denote the graph with vertex set $V(G)$ and edge set $F$. A set $F \subseteq E(G)$ such that $G[E(G) \setminus F]$ is {\Konig} is called a {\Konig} edge deletion set of $G$. 

The main focus of this paper is the following problem:

\defproblem{\KED}{An undirected graph $G = (V,E)$ and an integer $k$}{Does there exist $F \subseteq E$ such that $|F| \leq k$ and $G[E \setminus F]$ is a {\Konig} graph?}

\subsection{Parameterized Complexity} A parameterized problem $\Pi$ is a subset of 
$ \Gamma \times N$ where $\Gamma$ is a finite alphabet. An instance of a parameterized problem is a tuple
$(x, k)$, where $x$ is a classical problem instance, and $k$ is called the parameter.

A central notion in parameterized complexity is {\it fixed-parameter tractability} (FPT). 
A parameterized problem $\Pi$ is in FPT if there is an algorithm that
takes an instance $(x, k)$ and decides if $(x, k) \in \Pi$ in time $f(k) |x|^{O(1)}$.
Here, $f$ is an arbitrary function of $k$. Such an algorithm is called a fixed-parameter
algorithm (in short, an FPT algorithm).

There is an associated hardness theory with the notion of a {\it parameterized reduction} and a hierarchy of complexity classes. 
For two parameterized problems $\Pi _1$ and $\Pi _2$, a parameterized reduction from $\Pi _1$ to $\Pi _2$ is an algorithm 
that takes an input $(x,k)$ of $\Pi _1$ and converts in $f(k) |x|^{O(1)}$ time into an instance $(y, k')$ of 
$\Pi_2$ such that $(x,k) \in \Pi _1$ if and only if $(y, k') \in \Pi _2$. Here $f$ is some function of $k$ and $k'$ is 
$g(k)$ for some function $g$ of $k$.

For our purpose, it suffices to know that the independent set problem, which asks whether a given undirected graph has an independent set of size $k$, is a canonical complete problem for the class $W[1]$. Here $k$ is the parameter associated with the problem. See ~\cite{cygan2015parameterized} for more details on parameterized complexity.

\subsection{Vertex Cover Linear Program}
\label{NT}
The vertex cover problem can be expressed as an integer linear program where an optimum solution of the integer linear program would correspond to a minimum vertex cover. The variables of the linear program take values in $\{0,1\}$. However, we can relax this constraint on the variables and allow them to instead take on any real value in the interval $[0,1]$. This is called the LP relaxation of the vertex cover problem (also known as the fractional vertex cover problem)~\cite{lovasz2009matching}. It is defined as follows.

\[{\vc}_f(G) = min\{1^{\top} y :  \forall\{u,v\} \in E(G), y_u + y_v \geq 1,  \text{ and } \forall v \in V(G), y_v \in [0,1]\}\]

While finding the minimum vertex cover problem is NP-hard, finding the optimum solution for the fractional version can be done 
in polynomial time.

The dual of the vertex cover linear program is the LP relaxation of the maximum matching problem and is given as follows.

\[\mu_f(G) = max\{1^\top x : \forall v \in V(G) x(\delta(v)) \leq 1, \text{ and } \forall e \in E(G), x_e \in [0,1]\}\]

where $\delta(v)$ denotes the set of edges incident to a vertex $v$ and for $F \subseteq E$, $x(F) = \sum\limits_{e \in F} x_e$.

As the vertex cover linear program is a relaxation of a minimization problem, the optimal solution to the vertex cover linear program is always less than or equal to the size of the minimum vertex cover. Similarly, the fractional maximum matching size is always greater than or equal to the size of the maximum matching. Since $\mu_f$ and $vc_f$ are duals of each other, their optimal values are the same. Hence we have the following inequality.

\[\mu(G) \leq \mu_f(G) = {\vc}_f(G) \leq {\vc}(G)\]

For {\Konig} graphs, since $\mu(G) = {\vc}(G)$, the above inequalities collapse into a single value.

The fractional versions of these problems have been widely studied~\cite{NT75,lovasz2009matching,pulleyblank1987fractional,balas1981integer,nemhauser1974properties,trotter1973solution,picard1977integer}. 
For the vertex cover LP, it is known that there always exists an optimum solution where every variable takes values in $\{0,\frac{1}{2},1\}$~\cite{nemhauser1974properties,trotter1973solution}.
For an optimum solution to the vertex cover LP, let $S_1$ and $S_0$ be the set of vertices that get assigned $1$ and $0$ respectively. One may observe that there is a matching across the edges $(S_0,S_1)$ that saturates $S_1$. Suppose not, then by Hall's Theorem, there exists a subset $A \subseteq S_1$ such that $|N(A) \cap S_0| < |A|$. Now we may construct a smaller feasible solution by setting the vertices in $A \cup (N(A) \cap S_0)$ to $\frac{1}{2}$ and keeping the value of the other vertices the same. It is easy to see that this is feasible and has smaller value.
It is also known that we can find an LP solution that minimizes the number of vertices assigned $\frac{1}{2}$ in polynomial time. The following theorem summarizes these statements.

\begin{theorem} [Nemhauser-Trotter Theorem] \label{thm:nem_trotter} \cite{NT75,picard1977integer} 
Given a graph $G$, there exists a polynomial time algorithm that computes an optimum solution to the LP of Vertex Cover of $G$ such that
\begin{enumerate}
\item The values assigned by the solution are in $\{0,\frac{1}{2},1\}$. Let the set of vertices that get values $0, \frac{1}{2}, 1$ be $S_0, S_{\frac{1}{2}}$ and $S_1$ respectively.
\item The cardinality of $S_{\frac{1}{2}}$ is the minimum among all such solutions.
\item There exists a matching between $S_0$ and $S_1$ that saturates $S_1$.
\end{enumerate}
\end{theorem}



\subsection{Properties of {\Konig} Graphs}

We note the following observations that follow from the definition of {\Konig} graphs.

\begin{obs} \label{obs:konig_mm}
	Given a {\Konig} graph $G$, a maximum matching $M$ of $G$ and a minimum vertex cover $S$ of $G$, the following holds:
	\begin{enumerate}
	\item For any edge $e \in M$, we have $|e \cap S| = 1$.
	That is, for any edge in $M$, exactly one of its endpoints is in $S$, and

	\item Let $U$ be the set of unsaturated vertices of $M$, then $S \cap U = \emptyset$.
	\end{enumerate}
\end{obs}

From the above observation it follows that for every minimum vertex cover $S$ and maximum matching $M$ of a {\Konig} graph $G$, $M$ lies across the edges $(S,V(G) \setminus S)$ and $M$ saturates $S$. Moreover, if there exists a maximum matching $M$ and minimum vertex cover $S$ such that $M$ saturates $S$ and $M$ lies across $(S,V(G)\setminus S)$, then $|M| = |S|$ and hence $G$ must be {\Konig}. So we have the following characterization of {\Konig} graphs.

\begin{lemma} \cite{Deming79,Stersoul1979ACO} \label{lem:konig_characterization}
A graph $G$ is {\Konig} if and only if for every minimum vertex cover $S$ of $G$, there exists a matching across $(S, V(G) \setminus S)$ that saturates $S$. 
\end{lemma}

As for any vertex cover $S$ and a matching $M$ of $G$, $|S| \geq |M|$, we also have the following equivalent characterization
using the fact that if there is a matching $M$ and a vertex cover $S$ such that $|M|=|S|$, then $M$ must be a maximum matching and $S$ must be a minimum vertex cover.
\begin{lemma} \cite{MRSSS11} \label{lem:konig_partition}
A graph $G=(V,E)$ is {\Konig} if and only if for some vertex cover $S$ of $G$, there exists a matching across $(S, V(G) \setminus S)$ that saturates $S$.
\end{lemma}

In passing, we observe the following interesting property of {\Konig} graphs that we couldn't find in literature. 
Though we don't use it in the rest of our paper, we feel that this can be of independent interest.
\begin{lemma} \label{lem:konig_alternate}
A graph $G$ is {\Konig} if and only if ${\vc}_f(G) = {\vc}(G)$, that is, its vertex cover number is equal to its fractional vertex cover number.
\end{lemma}
\begin{proof}
If $G$ is {\Konig}, then ${\vc}_f(G) = {\vc}(G)$ as observed in Subsection~\ref{NT}. To prove the converse, find the optimum solution to the vertex cover linear program (that is, the fractional vertex cover) of $G$ that satisfies the properties in 
Theorem~\ref{thm:nem_trotter}. 

By definition, $S_0$ induces an independent set.  
As ${\vc}_f(G) = {\vc}(G)$, by the second property of the theorem, $S_{1/2} = \emptyset$ as there exists an optimum solution to the vertex cover linear program which is integral. Hence $(S_1, S_0)$ is a partition of the vertex set and $S_1$ is a vertex cover. By the third property of the theorem, there exists a matching across $(S_1, S_0)$ that saturates $S_1$. Hence the lemma follows from Lemma~\ref{lem:konig_partition}.
\end{proof}

\section{W[1]-hardness of {\Konig} Edge Deletion}



In this section, we will prove the following theorem:

\begin{theorem}
\label{thm:whardness-ked}
{\KED} is {\Wonehard}.
\end{theorem}

\noindent\textbf{The Reduction.} We will reduce from the parameterized independent set problem which is known to be {\Wonehard}.

Let $G=(V,E)$ be a graph with vertex set $V = \{v_1,v_2,...,v_n\}$ in which we want to find a $k$-sized independent set. Let, without loss of generality, $k < n/2$ (for otherwise $n \leq 2k$ and the 
independent set instance can be trivially solved in fixed-parameter-tractable time).
We will construct $G'$ as follows:

Initially, $G'$ is a copy of $G$. For every vertex $v_i \in V$, we add a pendant vertex $p_i$ to $G'$ such that $(v_i,p_i)$ is an edge. Let $P$ be the set of all such pendant vertices $p_i$. 


Add a set of $2k$ vertices $C = \{c_1,...,c_{2k}\}$ to $G'$ such that each $c_i$ is adjacent to every vertex in $V \cup P$ and $G'[C]$ is an edgeless graph.

This ends the construction. Formally, $G'=(V',E')$ is given as follows: \[V' = V \cup P \cup C\] and 
\[E' = E \cup \{(v_i,p_i) \mid 1 \leq i \leq n\} \cup \{(c,v),(c,p) \mid c \in C, v \in V, p \in P\}\]

Note that $G'$ has $n' = 2n + 2k$ vertices. See Figure~\ref{fig:reduction}.

We first prove the following claims about {\Konig} edge deletion sets of $G'$.
\begin{claim}
\label{claim:konig_vc_upper}
Let $F$ be a {\Konig} edge deletion set of $G'$ and let $S$ be a minimum vertex cover of $G'[E' \setminus F]$, 
then $|S| \leq n+k$.
\end{claim}
\begin{proof}
Suppose that $|S| > n+k$, then $|V' \setminus S| < n+k < |S|$. Hence there can not be a matching across the cut $(S,V'\setminus S)$
in $G'[E' \setminus F]$ that saturates $S$. Hence by Lemma \ref{lem:konig_characterization} $G'[E'\setminus F]$ is not {\Konig}, 
a contradiction to the definition of $F$.
\end{proof}


\begin{claim} \label{claim:redn_vc_size}
    Let $F$ be a {\Konig} edge deletion set of $G'$ of size at most $k$ and let $R_F = \{v_i \in V \mid (v_i,p_i) \notin F\}$. Then 
    \begin{enumerate}
        \item For any minimum vertex cover $S$ of $G'[E'\setminus F]$, $C \subseteq S$.
        \item There exists a minimum vertex cover $S$ of $G'[E'\setminus F]$ such that $P \subseteq V'\setminus S$.
        \item For any minimum vertex cover $S$ of $G'[E'\setminus F]$, $|S| \geq 2k + |R_F|$.
    \end{enumerate}
\end{claim}
\begin{proof}
Let $S$ be a minimum vertex cover of $G'[E'\setminus F]$, and let $c$ be a vertex in $C$. Since $c$ is adjacent to every vertex in $P$ and $V$, it has $2n$ neighbours in $G'$, and hence has at least $2n-k$ neighbours in $G'[E'\setminus F]$. 
Suppose that $c$ is not in $S$, then all its neighbours must be in $S$. However then $|S| \geq 2n-k > n+k$ (as $k < n/2$).
By Claim \ref{claim:konig_vc_upper}, this contradicts the fact that $F$ is {\Konig} edge deletion set. 
Hence every vertex $c \in C$ must be in $S$.
    
Let $S$ be a minimum vertex cover of $G'[E'\setminus F]$, and suppose there exists a vertex $p_i \in P$ such that 
$p_i$ is also in $S$. Since we know that all vertices of $C$ are in $S$, the only possible edge that $p_i$ might uniquely cover is the edge $(v_i,p_i)$. Consider the set $S' = S\setminus \{p_i\} \cup \{v_i\}$. Clearly $S'$ is also a vertex cover of $G'[E'\setminus F]$ of the same cardinality. As we can do this for every $p_i \in P$, we have that $P \subseteq V'\setminus S'$ proving the claim.

Let $S'$ be such a minimum vertex cover that contains no vertex of $P$.
Then all $v_i$ such that $(v_i,p_i) \notin F$ must be in $S'$ to cover the edge $(v_i,p_i)$. 
By definition of $R_F$, the number of these vertices is $|R_F|$. Also since $C \subseteq S'$, the size of $S'$ must be at least $|C| + |R_F| = 2k+|R_F|$.
\end{proof}

Let $F$ be a {\Konig} edge deletion set of $G$, and let $R_F$ be as defined in Claim~\ref{claim:redn_vc_size}, and let $S$ be a minimum vertex cover of $G'[V'\setminus F]$.
Then $|R_F| \geq n - |F|$ by the definition of $R_F$, and $|S| \leq n+k$ from Claim~\ref{claim:konig_vc_upper}. Furthermore if $|F| \leq k$, then we have from 
Claim~\ref{claim:redn_vc_size}, that $n+k \geq |S| \geq 2k + n - |F|$ and so it follows that $|F|=k$. Thus we have the following corollary.

\begin{corollary}
\label{cor:konigsize}
For any {\Konig} edge deletion set $F$ of $G'$, $|F| \geq k$.
\end{corollary}

Now we are ready to prove the correctness of the reduction.

\begin{claim} \label{claim:is_to_ked}
If $G$ has an independent set of size $k$ then $G'$ has a {\Konig} edge deletion set of size $k$.
\end{claim}
\begin{proof}
Let $I$ be the independent set of size $k$ in $G$.

Construct $V_1 = (V \setminus I) \cup C$ and $V_2 = P \cup I$ and let $F = (V_2, V_2)$. 
Clearly, $V_1$ and $V_2$ partition $V$. Since $I$ and $P$ are independent, there are no edges in $(P,P)$ and $(I,I)$ and so all edges of $(V_2,V_2)$ are in $(P,I)$. There are exactly $k$ many such edges because every vertex in $P$ is a pendant vertex and $|I|=k$. So $|F|= |(V_2,V_2)| = k$. We claim that $G'[V'\setminus F]$ is a {\Konig} graph. 
As $G'[V_2\setminus F] $ is independent, $V_1$ is a vertex cover of $G'[V'\setminus F]$. By Lemma~\ref{lem:konig_partition}, it sufficies to show that there is a matching $M$ across $(V_1, V_2)$ saturating $V_1$ in $G'[V'\setminus F]$.

Construct the matching $M$ as follows. Every vertex in $V \setminus I$ gets paired with its pendant vertex in $P$. 
Vertices in the first half of $C$ (that is, $c_1,\ldots,c_k$) get paired with vertices in $I$ (arbitrarily). Vertices in the second half of $C$ (that is, $c_{k+1}, \ldots , c_{2k})$ get paired with the pendant vertices corresponding to vertices of $I$ 
(arbitrarily). Because every vertex in $C$ is fully adjacent to vertices in $V \cup P$, such a pairing is always possible 
(see Figure~\ref{fig:reduction} for an illustration).
\end{proof}

\begin{figure}[t]
\centering
    \includegraphics[scale=0.35]{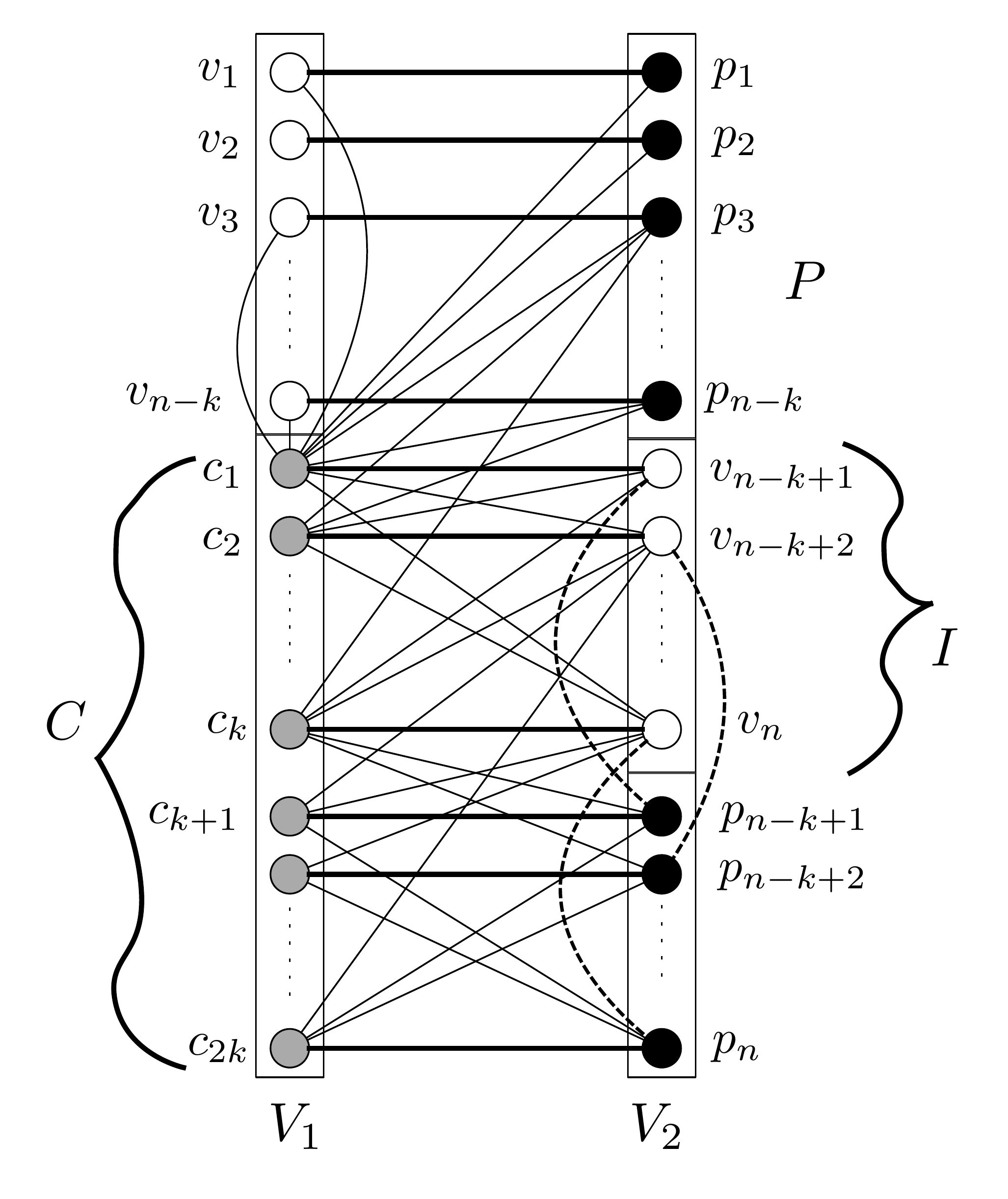}
    \caption{Illustration of the reduction where $I = \{v_{n-k+1},v_{n-k+2},...,v_n\}$ that is an independent set.
    Edges in $M$ are illustrated by thick line segments of black color, and the edges in $F$ are illustrated by `dotted' line segments (for example, the edge $(v_n,p_n)$).}
\label{fig:reduction}
\end{figure}

\begin{claim} \label{claim:ked_to_is}
If $G'$ has a {\Konig} edge deletion set of size $k$, then $G$ has an independent set of size $k$.
\end{claim}
\begin{proof}
Let $F$ be the {\Konig} edge deletion set of size $k$. Let $S'$ be the minimum vertex cover of $G'[E' \setminus F]$ given by Claim \ref{claim:redn_vc_size} such that $C \subseteq S'$ and $P \subseteq V'\setminus S'$. Suppose that there exists an edge $e \in (S',S') \cap F$ or an edge $e \in (S',V'\setminus S') \cap F$ then we can add $e$ back into the graph (by deleting from $F$) and the resulting graph will still remain {\Konig} by Lemma ~\ref{lem:konig_partition}. However, this gives us a {\Konig} edge deletion set of size $k-1$ contradicting Corollary~\ref{cor:konigsize}. 
Hence $F \subseteq (V' \setminus S', V' \setminus S')$. Moreover, since $V'\setminus S'$ is an independent set in $G'[E'\setminus F]$, $F$ must delete all edges in $V'\setminus S'$ and so $F \supseteq (V' \setminus S', V' \setminus S')$. Hence we have that 
$F = (V'\setminus S',V'\setminus S')$.

Since $G'[E' \setminus F]$ is {\Konig}, 
by Claim \ref{claim:konig_vc_upper}, $|S'| \leq n+k$ and hence $|V'\setminus S'| \geq n+k$. All vertices in $P$ are in $V' \setminus S'$, and so that accounts for $n$ vertices. Since $C \cap (V'\setminus S') = \emptyset$, all the other at least $k$ vertices in $V'\setminus S'$ must come from $V$. 
For each such vertex $v_i$ in $V \cap (V' \setminus S')$, there exists an edge of the form $(v_i,p_i)$ in $(V'\setminus S',V'\setminus S') = F$. 
As $|F|=k$, it follows that exactly $k$ vertices from $V$ are in $V'\setminus S'$. 
If there is an edge between two vertices in $V \cap (V'\setminus S')$, then that edge must be in $F$ making $|F| > k$, a contradiction. Hence the $k$ vertices in $V \cap (V'\setminus S')$ form a $k$-sized independent set and since they are from $V$, they also form a $k$-sized independent set in $G$.
\end{proof}

The Claims \ref{claim:is_to_ked} and \ref{claim:ked_to_is} complete the proof of Theorem \ref{thm:whardness-ked}.

Note that $G'$ has a perfect matching (as shown in Figure~\ref{fig:reduction}). We have the following corollary:

\begin{theorem}
{\KED} is {\Wonehard} even when the input graph has a perfect matching.
\end{theorem}

\section{{\Konig} Edge Deletion disjoint from Matching}

As the goal of the {\Konig} edge deletion problem is to delete edges to make the vertex cover size equal to the maximum matching, a natural question is whether the minimum {\Konig} edge deletion set should be allowed to delete edges of a maximum matching. This is a strong restriction since there are some graphs in which the minimum {\Konig} edge deletion set must pick some edges from a maximum matching of the graph. 
See Figure \ref{fig:counterexample} for an example of such a graph.  In the figure, the matching shown is the only maximum matching of the graph and all the minimum {\Konig} edge deletion sets of the graph intersect this matching (shown is a particular minimum {\Konig} edge deletion set).
This motivates the following variant of the {\Konig} edge deletion problem.

\begin{figure}
	
\begin {center}
\begin {tikzpicture}
\begin{scope} [vertex style/.style={draw,
                                       circle,
                                       minimum size=3pt,
                                       inner sep=0pt,
                                       outer sep=0pt,
										text depth=10pt
}] 
	\path	 
		(0.5,0) coordinate[vertex style](a)
		(-0.5,0) coordinate[vertex style](b)
		(0.5,1) coordinate[vertex style](c)
		(0.5,2) coordinate[vertex style](d)
		(-0.5,1) coordinate[vertex style](e)
		(-0.5,2) coordinate[vertex style](f)
		(0.5,-1) coordinate[vertex style](t1)
		(1.5,-1) coordinate[vertex style](t2)
		(-0.5,-1) coordinate[vertex style](t3)
		(-1.5,-1) coordinate[vertex style](t4)
		(1.5,0) coordinate[vertex style](t5)
		(1.5,1) coordinate[vertex style](t6)
		(-1.5,0) coordinate[vertex style](t7)
		(-1.5,1) coordinate[vertex style](t8)

	   ; 

\end{scope}

\begin{scope} [edge style/.style={draw=black}]
	\draw[edge style] (a)--(c);
	\draw[edge style] (b)--(e);
	\draw[edge style] (a)--(t1);
	\draw[edge style] (a)--(t2);
	\draw[edge style] (b)--(t3);
	\draw[edge style] (b)--(t4);
	\draw[edge style] (a)--(t5);
	\draw[edge style] (a)--(t6);
	\draw[edge style] (b)--(t7);
	\draw[edge style] (b)--(t8);

\end{scope}
\begin{scope} [edge style/.style={draw=black, line width=2}]
	\draw[edge style] (a)--(b);
	\draw[edge style] (t2)--(t1);
	\draw[edge style] (t3)--(t4);
	\draw[edge style] (t5)--(t6);
	\draw[edge style] (t7)--(t8);
\end{scope}
\begin{scope} [edge style/.style={draw=black, line width=2}]
	\draw[dashed] (e)--(f);
	\draw[dashed] (c)--(d);
\end{scope}
\end{tikzpicture}
\end{center}
\caption{A graph where the minimum {\Konig} edge deletion set must pick an edge from the maximum matching. The dotted edges denote the {\Konig} edge deletion set. The dotted edges plus the bold edges are the edges in the maximum matching}
\label{fig:counterexample}
\end{figure}

\defproblem{{\sc {\Konig} Edge Deletion disjoint from Matching}}{Graph $G = (V,E)$, a maximum matching $M$ of $G$ and integer $k$}{Does there exists $F \subseteq E$ such that $|F| \leq k$ and $G[E \setminus F]$ is {\Konig} and $F \cap M = \emptyset$?}

A similar variant was studied with respect to edge deletion to stable graphs. A graph is said to be stable if the optimum fractional maximum matching of the graph is equal to the size of its maximum matching. 
Note that due to Lemma \ref{lem:konig_alternate}, {\Konig} graphs can be characterized as the class of graphs for which the optimal fractional vertex cover is equal to the minimum vertex cover size, thus {\Konig} edge deletion and stable edge deletion are very similar problems. The problem of deleting $k$ non-matching edges (for a particular maximum matching) such that the resulting graph becomes stable was looked at by Bock et. al. where they prove that this variant of the problem remains NP-hard~\cite{bock2015finding}. 

We show that our variant of {\Konig} edge deletion problem remains NP-hard, but fixed-parameter tractable when parameterized by the solution size.
Therefore, it seems that the W-hardness of {\Konig} edge deletion problem lies in the cases where we might want to 
delete edges that reduce the size of the maximum matching. 

To prove our results, we show that the problem is equivalent (in a parameter preserving sense) to
the following problem known as {\sc Almost-2-SAT}.

\defproblem{{\sc Almost-2-SAT}}{A 2-CNF-SAT formula $F$ and an integer $k$}{Does there exist at most $k$ clauses of $F$ whose deletion results in a satisfiable formula?}

The {\sc Almost-2-SAT} problem was studied previously~\cite{garey2002computers} where the authors proved that it is NP-hard. It was later shown to be fixed-parameter-tractable with a $O^*(15^k)$ algorithm~\cite{RS09}\footnote{$O^*$ notation ignores polynomial factors}. 
The current best algorithm runs in $O^*(2.31^k)$ time~\cite{LNRRS14}.

\subsection{NP-hardness}

\begin{theorem}
\label{thm:ked-disjoint-from-mm-nphard}
The {\Konig} Edge Deletion disjoint from Matching problem is NP-hard.
\end{theorem}
\begin{proof}
We will reduce from {\sc Almost-2-SAT}. We are given an {\sc Almost-2-SAT} instance $(V,C,k)$ where $V$ is the set of variables and $C$ is the set of clauses. 
First we will preprocess the instance to remove any singleton clauses. We add two new variables $f$ and $x$. We add $k+1$ copies of the clauses $(\neg f \lor x)$ and $(\neg f \lor \neg x)$. The only way to satisfy these clauses is to set $f$ to $0$. For every clause $c_i \in C$ consisting of a single literal $l$. We will replace it with the clause $(l \lor f)$. Since $f$ is forced to be false, this forces the algorithm to either set $l$ to true or to delete the clause.

After the above preprocessing (which can be done in polynomial time), we know that all clauses have size exactly $2$. We will construct a graph $G'$ as follows:

For every literal $a$ in $V$, add a vertex labelled $v_a$. For every clause $(a \lor b) \in C$, where $a,b$ are literals, add an edge between vertices $v_a$ and $v_b$. Additionally, for every variable $a$, add an edge between $v_a$ and $v_{\neg a}$. This ends the construction of $G'$. Formally, $G'=(V',E')$ is given as follows:

\[V' = \bigcup\limits_{a \in V} \{v_a, v_{\neg a}\}\]

\[E' = \{(v_a,v_{\neg a}) \mid a \in V\} \cup \{(v_a,v_b) \mid (a \lor b) \in C\}\]

We also need to give a maximum matching $M$ as input. Let $M = \{(v_a,v_{\neg a}) \mid a \in V\}$. The size of any matching can not exceed half the number of vertices in the graph. Since there are $2|V|$ vertices in $G'$ and the size of $M$ is $|V|$, $M$ is a maximum matching.

We input the instance $(G',M,k)$ to the {\Konig} Edge Deletion disjoint from Matching problem.

\begin{claim}
	If $(V,C,k)$ is a {\sf YES} instance for {\sc Almost-2-SAT} then $(G',M,k)$ is a {\sf YES} instance for {\Konig} Edge Deletion disjoint from Matching.
\end{claim}
\begin{proof}
	Since $(V,C,k)$ is a {\sf YES} instance for {\sc Almost-2-SAT}, there must exist a subset of clauses $D \subseteq C$ such that $|D| \leq k$ and the 2-SAT instance $(V,C\setminus D)$ is satisfiable. Define $F = \{(v_a,v_b) \mid (a \lor b) \in D\}$, these are the edges of $G'$ that correspond to the deleted clauses in $D$. It should be clear that $|F| \leq k$ and $F \cap M = \emptyset$. We will prove that $F$ is a {\Konig} edge deletion set of $G'$.

	Since $(V,C\setminus D)$ is satisfiable, there exists a satisfying assignment $f : V \to \{0,1\}$. Define $V_1 = \{v_a \mid f(a) = 1\} \cup \{v_{\neg a} \mid f(a) = 0\}$ and define $V_2 = \{v_{a} \mid f(a) = 0\} \cup \{ v_{\neg a} \mid f(a) = 1 \}$. That is to say, $V_1$ is the set of vertices that correspond to the literals assigned to $True$ in $f$ and $V_2$ are the set of vertices that correspond to the literals assigned to $False$ in $f$. By definition, $V_1$ and $V_2$ are partitions of $V'$. Also note that, every edge of $M$ intersects $V_1$ exactly once. We will prove that $V_2$ is an independent set in $G'[E' \setminus F]$. 
	Suppose there were an edge $(v_a,v_b)$ such that $v_a,v_b \in V_2$. Since $v_a,v_b \in V_2$, it must be that the literals $a$ and $b$ are assigned to $False$ in $f$, however then the clause $(a \lor b)$ is unsatisfied by $f$ and so it must be in $D$. Then it must be that $(v_a,v_b) \in F$. Hence $V_2$ is an independent set in $G'[E' \setminus F]$ which implies that $V_1$ is a vertex cover of $G'[E' \setminus F]$. 
	
	Therefore we have a partition of $V$ into $V_1$ and $V_2$ such that every edge in $M$ intersects $V_1$ exactly once and $V_1$ is a vertex cover. By Lemma \ref{lem:konig_partition} $G'[E \setminus F]$ is {\Konig} and $F$ is a {\Konig} edge deletion set of size at most $k$.
\end{proof}

\begin{claim}
	If $(G',M,k)$ is a {\sf YES} instance for {\Konig} Edge Deletion disjoint from Matching then $(V,C,k)$ is a {\sf YES} instance for {\sc Almost-2-SAT}.
\end{claim}
\begin{proof}
	Since $(G',M,k)$ is a {\sf YES} instance for {\Konig} Edge Deletion disjoint from Matching, there must be a set $F \subseteq E'$ such that $|F| \leq k$ and $F \cap M = \emptyset$ such that $G'[E'\setminus F]$ is {\Konig}. Define the set of clauses $D = \{(a \lor b) \mid (v_a,v_b) \in F\}$, it is clear that $|D| \leq k$. We will prove that $(V,C\setminus D)$ is satisfiable and hence $D$ is a solution to the {\sc Almost-2-SAT instance}.

	Let $S$ be a minimum vertex cover of $G'[E'\setminus F]$. Since $G'[E'\setminus F]$ is {\Konig}, by Observation \ref{obs:konig_mm}, every edge in $M$ intersects $V_1$ exactly once. Since every edge in $M$ is of the form $(v_a,v_{\neg a})$, it must be the case that for every $a \in V$, exactly one of $v_a$ or $v_{\neg a}$ is in $V_1$.

	To prove that $(V,C\setminus D)$ is satisfiable we give a truth assignment of the variables $f : V \to \{0,1\}$. Define $f(a) = 1$ if $v_a \in V_1$ and $f(a) = 0$ otherwise. Since exactly one of $v_a$ or $v_{\neg a}$ is in $V_1$, $f(a)$ is well-defined. Suppose $(a \lor b) \in C\setminus D$ is a clause. Since $(v_a,v_b)$ is an edge and $V_1$ is a vertex cover of $G'[E'\setminus F]$, $v_a \in V_1$ or $v_b \in V_1$ and hence $f(a) = 1$ or $f(b) = 1$ and hence $(a \lor b)$ is satisfied.

	Hence every clause in $(V,C\setminus D)$ is satisfied and $D$ is a solution to the {\sc Almost-2-SAT} instance.
\end{proof}

Combining the above two claims, 
we have the proof of Theorem \ref{thm:ked-disjoint-from-mm-nphard}.
\end{proof}

Note that in the above reduction, the parameter $k$ is preserved. That is to say, an instance of {\sc Almost-2-SAT} with parameter $k$ gets mapped to an instance of {\Konig} Edge Deletion disjoint from Matching with parameter $k$. Hence we have the following corollary.

\begin{corollary} \label{cor:ppr}
	{\sc Almost-2-SAT} has a parameter-preserving reduction to {\Konig} Edge Deletion disjoint from Matching.
\end{corollary}

\subsection{Fixed-parameter-tractability}

In this subsection, we will prove the fixed-parameter tractability of {\Konig} Edge Deletion disjoint from Matching.

\begin{theorem}
\label{thm:ked-disjoint-from-matching-fpt}
There is a fixed-parameter-tractable algorithm to find a minimum {\Konig} Edge Deletion set disjoint from Matching.
\end{theorem}
\begin{proof}
We will reduce the problem to {\sc Almost-2-SAT}. Given a graph $G$ with edge set $E$, and a maximum matching $M$ of $G$, for every edge $(u,v) \in G$, we create a clause $(u \lor v)$ to assert that one of $u$ or $v$ must be picked. For every edge $(u,v) \in M$, we create a clause $(\neg u \lor \neg v)$ to assert that both endpoints of this edge can not be picked. For every vertex $v$ that is unsaturated, we add the singleton clause $\neg v$ to assert that unsaturated vertices can not be picked. That is, we create the following 2-SAT instance:

\[\bigwedge\limits_{(u,v)\in M} (\neg u \lor \neg v) \bigwedge\limits_{(u,v)\in E} (u \lor v) \bigwedge\limits_{v \in U} \neg v\]

where $U$ is the set of vertices that are unsaturated by $M$. Furthermore, for every $(u,v) \in M$, we make $k+1$ duplicates of $(u \lor v)$ and $(\neg u \lor \neg v)$. We also make $k+1$ duplicates of the clauses of the form $\neg v$ (corresponding to unsaturated vertices).  We refer to the set of clauses as $C$. $C$ will be the input to the {\sc Almost-2-SAT} algorithm with parameter $k$.

\begin{claim}
\label{claim:if-yes-then-ked}
If the formula is a {\sf YES} instance then there exists a {\Konig} edge deletion set disjoint from $M$ of size $k$.
\end{claim}
\begin{proof}
Let $D$ be the set of deleted clauses. Any clause in $D$ which corresponds to an edge in $M$ occurs $k+1$ times in $C$, and since $|D| \leq k$, there will be at least one copy of the clause remaining in $C \setminus D$. This means that the removal of such clauses does not effect the satisfiability of $C \setminus D$ and hence we can assume that $D$ does not contain any clause corresponding to edges of $M$. The same argument can be used for clauses of the form $\neg v$. Hence we will assume that $D$ contains only clauses of the form $(u \lor v)$ where $(u,v) \in E \setminus M$.

For every deleted clause of the form $(u \lor v)$, $(u,v) \in E \setminus M$, we will remove the edge $(u,v)$ from $G$. Let the resulting graph be $G'$, we claim that $G'$ is a {\Konig} Graph. To prove this we provide a vertex cover $S$ of $G'$ such that every edge in $M$ is covered exactly once, that is, no edge has both endpoints in the vertex cover. Since the SAT instance $C \setminus D$ is satisfiable, there exists a satisfying assignment $f : V \to \{0,1\}$. We set the vertex cover $S$ to be the subset of variables that are set to true in $f$. This is indeed a vertex cover of $G'$ since for every edge $(u,v) \in E \setminus F$, we have added a clause $(u \lor v)$ that was not deleted by $D$. Hence, the clause was satisfied by $f$ and one of either $u$ or $v$ was set to true. Therefore, one endpoint of $(u,v)$ would be in $S$. Moreover, for every edge $(u,v)$ in $M$, we added a clause $(\neg u \lor \neg v)$ that is not deleted by $D$, this means both $u$ and $v$ can not be set to true in $f$ and hence both endpoints of $(u,v)$ are not in $S$. For every edge in $M$ exactly one of its endpoints is in $S$ and hence $|M| = |S|$ and hence $G'$ is {\Konig}.
%
\end{proof}

\begin{claim}
\label{claim:if-ked-them-yes}
If there exists a {\Konig} edge deletion set disjoint from $M$ of size $k$ then the formula is a {\sf YES} instance.
\end{claim}
\begin{proof}
Let $F$ be such a {\Konig} edge deletion set. For every edge $(u,v) \in F$, we will delete the clause $(u \lor v)$ from $C$ to obtain $C'$. We claim that $C'$ is satisfiable. The graph $G[E \setminus F]$ is {\Konig} and hence it has a vertex cover $S$ such that for every edge in $M$ both endpoints are not in $S$. We will construct an assignment $f$ as follows: For every vertex in $v \in G[E\setminus F]$ we set $f(v) = 1$ if $v \in S$ and $f(v) = 0$ otherwise. Since $S$ is a vertex cover, all clauses of the form $(u \lor v)$ are satisfied as either $u$ or $v$ will be set to true. Since, by Observation \ref{obs:konig_mm}, for every matching edge $(u,v) \in M$ there is exactly one endpoint in $S$, all clauses of the form $(\neg u \lor \neg v)$ will be satisfied as at least one of $u$ or $v$ will be set to false. Moreover, by Observation \ref{obs:konig_mm}, $S$ does not contain any unsaturated vertices, so all clauses of the form $\neg v$ will also be satisfied.
\end{proof}
\noindent
Combining Claims~\ref{claim:if-yes-then-ked} and~\ref{claim:if-ked-them-yes} and the fixed-parameter tractability of the {\sc Almost-2-SAT} problem, we have the proof of Theorem~\ref{thm:ked-disjoint-from-matching-fpt}.
\end{proof}

The above reduction also preserves the parameter $k$ and along with Corollary \ref{cor:ppr}, we have the following corollary.

\begin{corollary}
	{\sc Almost-2-SAT} and {\Konig} Edge Deletion disjoint from Matching are fixed-parameter equivalent problems.
\end{corollary}

\section{Conclusion}
In most cases, edge deletion problems tend to be easier (and admit faster FPT algorithms) than the corresponding vertex versions.
For example, covering cycles by deleting minimum number of edges (feedback edge set problem) in undirected graphs is polynomial time solvable, while the related feedback vertex set problem is {\NPComp}.
Also, covering odd cycles by deleting edges in undirected graphs admits a more efficient FPT algorithm~\cite{PPW2018algo} although both the problems are {\NPComp}. Covering cycles by minimum number of arcs (feedback arc set) in tournaments has a sub-exponential algorithm~\cite{AlonLS09} while the vertex version of the problem (feedback vertex set) in tournaments is unlikely to admit sub-exponential algorithms under Exponential Time Hypothesis~\cite{KumarL16,DomGHNT10,IPZ01}.

Surprisingly we have shown here that that {\Konig} Edge Deletion problem is {\Wonehard} while {\Konig} vertex deletion problem was known to be FPT.
From the variant of the problem that is fixed-parameter-tractable, it appears that the W[1]-hardness of {\Konig} Edge Deletion lies in the case when edges in the maximum matching of the graph can also be deleted. 

For the variant where the edges of a given maximum matching can not be deleted, we have shown it to be equivalent to the
the {\sc Almost-2-SAT} problem by giving a parameter preserving reduction. This puts that variation of the problem, equivalent to 
other popular parameterized problems including {\sc Above Guarantee Vertex Cover}, {\sc Odd Cycle Traversal} and {\sc Edge Bipartization}~\cite{RS09} and so the results on parameterized complexity (including kernelization), approximation and exact exponential algorithms for these problems~\cite{LNRRS14} carry over for this problem as well. 

We conclude with an open problem. A graph is said to be {\it stable} if the size of its maximum matching is equal to the size of the maximum fractional matching (i.e. the optimum value of the maximum matching linear program).
It is known that finding a minimum number of vertices whose removal results in a stable graph is solvable in polynomial time, but  the related edge version is {\NPComp}~\cite{bock2015finding}. Is the problem of deleting at most $k$ edges to obtain a stable graph fixed-parameter tractable? Note that the Edge Induced Stable Subgraph problem (in which one needs to find a set of $k$ edges that induces a stable subgraph) is fixed-parameter-tractable as the lower bounds for the size of an edge induced {\Konig} subgraph 
~\cite{konigkernel,MRSSS11} also apply to stable graphs, as {\Konig} graphs are stable.


\end{document}